\documentclass[a4paper,11pt]{article}

\usepackage[T1]{fontenc}
\usepackage[utf8]{inputenc}
\usepackage{authblk}

\usepackage{algorithm}
\usepackage{algorithmic}
\usepackage{mathrsfs}
\usepackage{amsmath,amsthm,amssymb}
\usepackage{cite}
\usepackage{txfonts}
\usepackage{graphics}
\usepackage{float}
\usepackage{epstopdf}
\usepackage{epsfig}
\usepackage{tikz}
\usepackage{caption}
\usepackage{subcaption}
\usepackage{url}
\usepackage{booktabs}

\title{On the Optimality of Tape Merge of Two Lists with Similar Size}
\date{}
\author[1,2]{Qian Li}
\author[1,2]{Xiaoming Sun}
\author[1,2]{Jialin Zhang}
\affil[1]{CAS Key Lab of Network Data Science and Technology, Institute of Computing Technology, Chinese Academy of Sciences, 100190, Beijing, China.}
\affil[2]{University of Chinese Academy of Sciences, Beijing, 100049, China.\\
  \texttt{\{liqian, sunxiaoming, zhangjialin\}@ict.ac.cn}}

\newtheorem{definition}{Definition}

\newtheorem{theorem}{Theorem}

\newtheorem{lemma}{Lemma}

\newtheorem{conjecture}{Conjecture}

\begin{document}

\maketitle
\begin{abstract}
The problem of merging sorted lists in the least number of pairwise comparisons has been solved completely only for a few special cases. Graham and Karp~\cite{taocp} independently discovered that the tape merge algorithm is optimal in the worst case when the two lists have the same size. In the seminal papers, Stockmeyer and Yao\cite{yao},  Murphy and Paull\cite{3k3},  and Christen\cite{christen1978optimality} independently showed when the lists to be merged are of size $m$ and $n$ satisfying $m\leq n\leq\lfloor\frac{3}{2}m\rfloor+1$, the tape merge algorithm is optimal in the worst case. This paper extends this result by showing that the tape merge algorithm is optimal in the worst case whenever the size of one list is  no larger  than 1.52 times the size of the other.  The main tool we used to prove lower bounds is Knuth's adversary methods~\cite{taocp}. In addition, we show that  the lower bound cannot be improved to 1.8 via  Knuth's adversary methods. We also develop a new inequality about Knuth's adversary methods, which might be interesting in its own right. Moreover, we design a simple procedure to achieve constant improvement of the upper bounds for $2m-2\leq n\leq 3m $.
\end{abstract}

\section{Introduction}

Suppose there are two disjoint linearly ordered lists $A$ and $B$:
$$a_1<a_2<\cdots<a_m$$
and
$$b_1<b_2<\cdots<b_n$$
respectively, where the $m+n$ elements are distinct.
The problem of merging them into one ordered list is one of the most fundamental algorithmic problems which has many practical applications as well as important theoretical significance. This problem has been extensively studied under different models, such as comparison-based model~\cite{taocp}, parallel model~\cite{Gavril}, in-place merging model~\cite{Huang1988} etc.
 In this paper we focus on the classical comparison-based model, where the algorithm is a sequence of pairwise comparisons. People are interested in this model due to two reasons. Firstly, it is independent to the underlying order relation used, no matter it is "$<$" in $\mathbf{R}$ or another abstract order relation.  Secondly, it is unnecessary to access the value of elements in this model. Such a restriction could come from security or privacy concerns where the only operation available is a zero-knowledge pairwise comparison which reveals only the ordering relation between elements.


The main theoretical question in this merge problem is to determine $M(m,n)$, the minimum number of comparisons which is always sufficient to merge the lists\cite{taocp}.
 Given any algorithm $g$\footnote{We only consider deterministic algorithms in this paper.} to solve the $(m,n)$ merging problem (i.e. where $|A|=m$ and $|B|=n$), let $M_{g}(m,n)$ be the number of comparisons required by algorithm $g$ in the worst case, then
$$M(m,n)=\min_{g}M_{g}(m,n).$$
An algorithm $g$ is said to be optimal on $(m,n)$ if $M_g(m,n)=M(m,n)$. By symmetry, it is clear that $M(m,n)=M(n,m)$. To much surprise, this problem seems quite difficult in general, and exact values are known for only a few special cases. Knuth determined the value of $M(m,n)$ for the case $m,n\leq 10$ in his book~\cite{taocp}. Graham~\cite{taocp} and Hwang and Lin~\cite{hwang1971optimal} completely solved the case $m=2$ independently. The case $m=3$ is quite a bit harder and was solved by Hwang~\cite{hwang1980optimal} and Murphy~\cite{Murphyreport}. M{\"o}nting solved the case $m=4$ and also obtained strong results about $m=5$~\cite{5n}. In addition, Smith and Lang~\cite{game} devised a computer program based on game solver techniques such as alpha beta search to compute $M(m,n)$. They uncovered many interesting facts including $M(7,12)=17$, while people used to believe $M(7,12)=18$.

Several different algorithms have been developed for the merge problem, among them \emph{tape merge} or \emph{linear merge} might be the simplest and most commonly used one.
In this algorithm, two smallest elements (initially $a_1$ and $b_1$) are compared, and the smaller one will be deleted from its list and placed on the end of the output list. Then repeat the process until one list is exhausted. It's easy to see that this algorithm requires $m+n-1$ comparisons in the worst case, hence $M(m,n)\leq m+n-1$. However, when $m$ is much smaller than $n$, it is obvious that this algorithm becomes quite inefficient. For example, when $m=1$, the merging problem is equivalent to an insertion problem and the rather different \emph{binary insertion} procedure is optimal, i.e. $M(1,n)=\lceil\lg(n+1)\rceil$.

One nature question is "{\it when is tape merge optimal?}". By symmetry, we can assume $n\geq m$ and define $\alpha (m)$ be the maximum integer $n (\geq m)$ such that tape merge is  optimal,  i.e.
$$\alpha (m)=\max\{n\in \mathbb{N}\ |\ M (m, n)= m+n-1,  n\geq m\}. $$
Assume a conjecture proposed by Knuth~\cite{taocp}, which asserts that $M(m,n+1)\leq M(m,n)+1\leq M(m+1,n)$, for $m\leq n$, is correct, it's easy to see tape merge is optimal if and only if $n\leq\alpha(m)$, and $\alpha(m)$ is monotone increasing.

Graham and Karp~\cite{taocp} independently discovered that $M(m,m)=2m-1$ for $m\geq 1$. Then Knuth~\cite{taocp} proved $\alpha(m)\geq 4$ for $m\leq 6$. Stockmeyer and Yao\cite{yao},  Murphy and Paull\cite{3k3},  and Christen\cite{christen1978optimality} independently significantly improved the lower bounds by showing $\alpha(m)\geq \lfloor\frac{3}{2}m\rfloor+1$, that is $M (m, n)=m+n-1$,  for $m\leq n\leq \lfloor\frac{3}{2}m\rfloor+1$.
On the other hand,  Hwang\cite{simple} showed that $M (m, 2m)\leq 3m-2$, which implies $\alpha(m)\leq 2m-1$. For $m\leq n\leq 2m-1$, the best known merge algorithm is tape merge algorithm. It is conjectured by Fernandez et al. \cite{fernandeztwo} that $\alpha(m)=\frac{1+\sqrt{5}}{2}m\pm o(m)$.

For general $n\geq m$, Hwang and Lin\cite{simple} proposed an in-between algorithm called \emph{binary merge}, which excellently compromised between binary insertion and tape merge in such a way that the best features of both are retained. It reduces to tape merge when $n\leq2m$, and reduces to binary insertion when $m=1$. Let $M_{bm}(m,n)$ be the worst-case complexity of this algorithm. They showed that
    $$M_{bm}(m,n)=m(1+\lfloor\lg\frac{n}{m}\rfloor)+\lfloor\frac{n}{2^{\lfloor\lg\frac{n}{m}\rfloor}}\rfloor-1.$$

Hwang and Deutsch\cite{hwang1973class} designed an algorithm which is optimal over all \emph{insertive algorithms} including binary merge, where for each element of the smaller list, the comparisons involving it are made consecutively. However,  the improvement for fixed $n/m$ over binary merge increases more slowly than linearly in $m$\cite{manacher1979significant}.
 Here we say that algorithm $A_1$ with complexity $M_{A_1}(m,n)$ is significantly faster for some fixed ratio $n/m$ than algorithm $A_2$ with complexity $M_{A_2}(m,n)$, if $M_{A_2}(m,n)-M_{A_1}(m,n)=\Omega(m)$. The first significant improvement over binary merge was proposed by Manacher\cite{manacher1979significant},  which can decrease the number of comparisons by $\frac{31}{336}m$ for $n/m\geq8$, and Thanh and Bui\cite{Animprovement} further improved this number to $\frac{13}{84}m$.  In 1978,  Christen\cite{christen1978improving} proposed an elegant algorithm, called \emph{forward testing and backward insertion},  which is better than binary merge when $n/m\geq3$ and saves at least $\sum_{j=1}^{k}\lfloor\frac{m-1}{4^j}\rfloor$ comparisons over binary merging, for $n\geq 4^km$. Thus it saves about $m/3$ comparisons when $n/m\rightarrow\infty$.  Moreover,  Christen's procedure is optimal for $5m-3\leq n\leq7m$, i.e. $M (m, n)=\lfloor (11m+n-3)/4\rfloor$.

On the lower bound side, there are two main techniques in proving lower bounds. The first one is the information theoretic lower bound $ I(m,n)=\lceil\lg{{m+n}\choose{m}}\rceil$.
 Hwang and Lin\cite{simple} have proved that
 $$I(m,n)\leq M(m,n)\leq M_{bm}(m,n)\leq I(m,n)+m.$$

The second one is called \emph{Knuth's adversary methods}~\cite{taocp}. The idea is that the optimal merge problem can be viewed as a two-player game with perfect information, in which the algorithm chooses the comparisons, while the adversary chooses(consistently) the results of these comparisons. It is easy to observe that $M(m,n)$ is actually the min-max value of this game. Thus a given strategy of the adversary provides a lower bound for $M(m,n)$. Mainly because of the consistency condition on the answers, general  strategies are rather tedious to work with. Knuth proposed the idea of using of "disjunctive" strategies, in which a splitting of the remaining problem into two disjoint problems is provided, in addition to the result of the comparison. With this restricted adversary, he used term $.M.(m,n)$ to represent the minimum number of comparisons required in the algorithm, which is also a lower bound of $M(m,n)$.  The detail will be specified in Section~\ref{section:definition}.

%

\subsection{Our Results}
In this paper, we first improve the lower bounds of $\alpha(m)$ from $\lfloor\frac{3}{2}m\rfloor+1$ to $\lfloor\frac{38}{25}m\rfloor$ by using Knuth's adversary methods.
\begin{theorem}\label{thm:38/25}
$M (m, n)=m+n-1$,  if $m\leq n\leq \frac{38}{25}m$.
\end{theorem}

We then show limitations of Knuth's adversary methods.
\begin{theorem}\label{thm:9/5}
$. M.  (m, n)<m+n-1$, if $n\geq9\lceil m/5\rceil$.
\end{theorem}

This means that by using Knuth's adversary methods, it's impossible to show $\alpha(m)\geq9\lceil m/5\rceil\approx\frac{9}{5}m$ for any $m$.

When $m\leq n\leq 3m$, binary merge is the best known algorithm, which reduces to tape merge for $n\in[m,2m]$ and gives $M_{bm}(m,2m+k)=3m+\lfloor k/2\rfloor-1$ for $k\in[0,m]$.
In this paper, we give improved upper bounds for $M (m, n)$ for $2m-2\leq n\leq 3m$. In particular, it also improves the upper bounds of $\alpha(m)$, that is, $\alpha(m)\leq 2m-3$ for $m\geq 7$.
\begin{theorem}\label{thm:upper_bound}

\begin{enumerate}
  \item[(a)] $M (m, 2m+k)\leq 3m+\lfloor k/2\rfloor-2 = M_{bm}(m,2m+k)-1$,  if $m\geq 5$ and $k\geq -1$.
  \item[(b)] $M (m, 2m-2)\leq 3m-4 = M_{bm}(m,2m-2)-1$,  if $m\geq 7$. That is $\alpha(m)\leq 2m-3$ for $m\geq 7$.
  \item[(c)] $M (m, 2m)\leq 3m-3 = M_{bm}(m,2m)-2$,  if $m\geq 10$.
\end{enumerate}
\end{theorem}

\subsection{Related work}\label{section:related work}

%

%
%


Besides the worst-case complexity, the average-case complexity  has also been investigated for merge problems. Tanner~\cite{Tanner} designed an algorithm which uses at most $1.06I(m,n)$ comparisons on average. The average case complexity of insertive merging algorithms as well as binary merge has also been investigated~\cite{FernandezdelaVega1998, fernandeztwo}.

Bui et al.~\cite{MAITHANH1986341} gave the optimal randomized algorithms for $(2,n)$ and $(3,n)$ merge problems and discovered that the optimal expected value differs from the optimal worst-case value by at most 1. Fernandez et al.~\cite{fernandeztwo} designed a randomized merging algorithm which performs well for any ratio $n/m$ and is significantly faster than binary merge for $n/m>(\sqrt{5}+1)/2\approx1.618$. More preciously, they showed that
\begin{equation}
         M_{F}(m,n)=\left\{
         \begin{aligned}
             & sn+(1+s)m,& &if& &1+s\leq n/m\leq2+s,&\\
             & 2\sqrt{mn},& &if& &2+s\leq n/m\leq 2r,&
         \end{aligned}
         \right.
     \end{equation}
     Where $s=(\sqrt{5}-1)/2\approx 0.618$ and $r=(\sqrt{2}-1+\sqrt{2}s)^2\approx1.659$.



    Nathan Linial~\cite{Linial} studied a more general problem where partial order relations are already known. He showed the information-theoretic lower bound is good, that is, an algorithms exists which merges $A$ and $B$ in no more than $(\lg(\sqrt{5}+1)/2)^{-1}\lg N_0$ comparisons, where $N_0$ is the number of extensions of the partial order on $A\cup B$. They also pointed out that this bound is tight, and the computation needed for finding the appropriate queries can be done in time polynomial in $m+n$.


     Sorting, merging, selecting and searching are always closely related to each other. Manacher et al.~\cite{Manacher2} used efficient merge algorithms to improve Ford-Johnson sorting algorithm, which was conjectured to be optimal for almost twenty years.  Linial and Saks~\cite{saks} observed that $M(m,n)$ is equivalent to the minimum number of pairwise comparisons to search an entry in a $m\times n$ matrix in which distinct entries are known to be increasingly ordered along rows and columns.  They also studied the generalized problem in monotone multi-dimensional arrays, and their result was further improved by Cheng et al.~\cite{Cheng08}. Ajtai et al.~\cite{Ajtai} considered the problem of sorting and selection with imprecise comparisons. The non-uniform cost model has also been investigated~\cite{Survey,Gupta:2001,Kannan,Zhiyi}, for example, Huang et al.~\cite{Zhiyi} studied the sorting problem where only a subset of all possible pairwise comparisons is allowed.
     For practical use, Brown and Tarjan~\cite{Tarjan} gave a merging procedure which runs in $O(I(m,n))$ time on a real computer.


\vspace{1em}
\textbf{Organization.  }
We introduce some notations and explain Knuth's adversary methods in Section~\ref{section:definition}. In Section~\ref{section:general_ineq} we provide some properties of Knuth's adversary methods  which will be used. In Section~\ref{section:lowerbound} we improve the lower bounds for $\alpha(m)$ via Knuth's adversary methods. Then we show limitations of this method in Section~\ref{section:upperbound2}. Section~\ref{section:upper_bound1} improves the upper bounds for $M(m,2m+k)$.
We conclude the paper  with some open problems in Section~\ref{section:conclusion}.

\section{Preliminaries}\label{section:definition}

In this section, we introduce Knuth's adversary methods and present some notations.

We use the notations $\lambda M\rho$ proposed by Knuth in this paper. The detailed definitions can be found in Knuth's comprehensive monograph~\cite{taocp}. Yao et al.~\cite{yao} gave an example to illustrate the use of that. Here, we briefly introduce the idea.

The basic idea of Knuth's  adversary methods is to restrict the possible adversary strategies. In general, the adversary can arbitrarily answer the comparison query from the algorithm as long as there are no contradictions in his answer. But in Knuth's adversary methods, after each comparison query between $a_i$ and $b_j$, the adversary is required to split each sorted list into two parts $A = A_1\cup A_2$ and $B = B_1 \cup B_2$ (Figure 1). The adversary guarantees that each element in $A_1$ or $B_1$ is smaller than any element in $A_2$ or $B_2$. It is also guaranteed that $a_i$ and $b_j$ are not in the same subproblem, i.e. neither $a_i\in A_1, b_j \in B_1$ nor $a_i\in A_2, b_j \in B_2$, thus, the comparison result between $a_i$ and $b_j$ is determined after the splitting.  Then the merge problem will be reduced to two subproblems $(A_1, B_1)$ and $(A_2, B_2)$ with different left or right constraints. For example, in case 2, the constraint for subproblem $(A_1, B_1)$ is a right constraint $b_l<a_k$ since $a_k\in A_2$ while $b_l\in B_1$.
\begin{figure}[!ht]
  \centering
      \begin{subfigure}[b]{0.3\textwidth}
        \includegraphics[width=\textwidth]{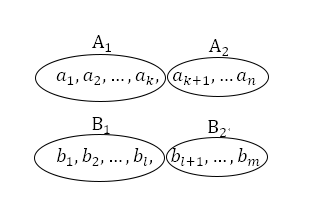}
        \caption{ Case 1}
      \end{subfigure}
       \begin{subfigure}[b]{0.3\textwidth}
        \includegraphics[width=\textwidth]{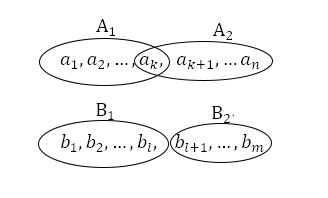}
        \caption{ Case 2}
      \end{subfigure}
       \begin{subfigure}[b]{0.3\textwidth}
        \includegraphics[width=\textwidth]{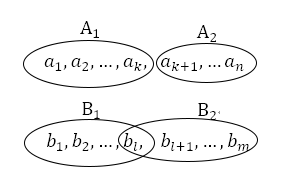}
        \caption{ Case 3}
      \end{subfigure}
  \caption{The Adversary's Splitting Strategies}
\end{figure}

Knuth introduced notation $\lambda M\rho$ to represent nine kinds of restrict adversaries, where $\lambda,\rho\in \{.,\backslash,/\}$ are the left and right constraint. In general, the constraint notation $'.'$  means no left (or right) constraint. Left constraint $\lambda=\backslash $ or $ /$ means that outcomes must be consistent with $a_{1}<b_{1}$ or $a_{1}>b_{1}$ respectively.  Similarly, right constraint $\rho=\backslash (/)$ means the outcomes must be consistent with $a_{m}<b_{n}$ (or $a_{m}>b_{n}$ respectively). Thus, merge problem $\lambda M\rho(A,B)$ will reduce to subproblem $\lambda M.(A_1, B_1)$ and $.M\rho(A_2, B_2)$ in case 1, to subproblem $\lambda M/(A_1, B_1)$ and $\backslash M\rho(A_2, B_2)$ in case 2, and to subproblem $\lambda M\backslash(A_1, B_1)$ and $/M\rho(A_2, B_2)$ in case 3. For convenience, we say the adversary adopts a \emph{simple} strategy if he splits the lists in the way of case 1, otherwise the adversary adopts a \emph{complex} strategy (case 2 or 3).
 There are obvious symmetries, such as $/M.  (m, n)=.M\backslash (m,n)=\backslash M.  (n, m)= .M/  (n, m)$, $/M/  (m, n)= \backslash M\backslash  (m, n)$, and $/M\backslash(m,n)=\backslash M/(n,m)$, which means we can deduce the nine functions to four functions: $.M.$, $/M.$, $/M\backslash$, and $/M/$. These functions can be calculated by a computer rather quickly, and the values for all $m,n\leq 150$ and the program are available in~\cite{table}.

Note $M(m,n)\geq.M.(m,n)$, but $M(m,n)$ is not equal to $.M(m,n).$ in general, since we restrict the power of adversary in the decision tree model by assuming there is a (unknown) division of the lists after each comparison. But this restrict model still covers many interesting cases.
For example, when $m\leq n\leq \lfloor\frac{3}{2}m\rfloor+1$~\cite{christen1978optimality,3k3,yao} or $5m-3\leq n\leq 7m$~\cite{christen1978improving}, $.M(m,n). = M(m,n)$.

Let $\lambda M_{i, j}\rho (m, n)$ denote the number of comparisons resulted from adversary's best strategy if the first comparison is $a_{i}$ and $b_{j}$, thus $\lambda M\rho (m, n)=min_{i, j}\lambda M_{i, j}\rho (m, n)$.

The following notation is also used in our paper.
\begin{definition}
Let $\overline{.M.}(m,n)$ be the difference of the number of comparisons required by tape merge in the worst case and $.M.(m,n)$, i.e.
$\overline{.M.}(m,n)\triangleq m+n-1-.M.(m,n).$
\end{definition}

\section{ Inequalities about Knuth's adversary methods}\label{section:general_ineq}

In this section, we list several inequalities about $\lambda M\rho$, which will be used in Section~\ref{section:lowerbound} and Section~\ref{section:upperbound2}.

\begin{lemma}\label{prop:yao}
For any $\lambda,\rho\in\{.,/,\backslash\}$, we have
\begin{enumerate}
  \item[(a)] $.M\rho(m,n)\geq \lambda M\rho(m,n)$;
  \item[(b)] $/M\rho (m, n)\leq. M\rho (m, n-1)+1$.
\end{enumerate}

\end{lemma}
\proof
\emph{Part (a)} is obvious, the adversary can perform at least as well on less restrictions. In \emph{Part (b)}, if the  first comparison is $a_{1}$ and $b_{1}$ for $/M\rho (m, n)$,  the adversary has to claim $a_{1}>b_{1}$, thus it reduces to $. M\rho (m, n-1)$. Therefore  we have $/M\rho (m, n)\leq/M_{1, 1}\rho(m,n)=1+. M\rho (m, n-1)$.
\endproof

In the following, we will show that if $.M.(m,n)=m+n-1$, then tape merge is optimal for any $(m',n')$ satisfying $m'\geq m$, $n'\leq n$ and $m'\leq n'$. That is
 \begin{lemma}\label{25384}
 For any $m,n\geq0$, $m+n\geq1$ and $m\leq n$, we have $. M.  (m+1, n)\geq. M.  (m, n)+1\geq. M.  (m, n+1)$ or
 $\overline{.M.}(m+1,n)\leq \overline{.M.}(m,n)\leq \overline{.M.}(m, n+1)$.
 \end{lemma}
In order to prove this lemma, We show the following lemma first:
\begin{lemma}\label{25383}
         $. M.  (m+1, n+1)\geq. M.  (m, n)+2$, for any $m,n\geq 0$ and $m+n\geq 1$.
     \end{lemma}
  \begin{proof}
The proof is by induction on $m$ and $n$.  The starting values for $m, n \leq 3$  can be easily checked in~\cite{taocp}.  Now suppose the theorem holds for any $m',n'$ satisfying $m'\leq m$, $n'\leq n$ and $m'+n'<m+n$, we then prove the case $(m,n)$. Note that our task is to design a strategy for the adversary for $(m+1,n+1)$.

 Suppose an algorithm begins by comparing $a_{i}$ and $b_{j}$,  where $i\leq m$,  $j=n+1$.  The adversary claims that $a_{i}<b_{j}$,  and follows the simple strategy,  yielding
$$. M_{i, j}.  (m+1, n+1)\geq1+. M.  (m, n)+. M.  (1, 1)\geq. M.  (m, n)+2$$

If $i=m+1$ and $j\leq n$,  the adversary claims that $a_{i}>b_{j}$,  and uses the simple strategy.  This leads to
$$. M_{i, j}.  (m+1, n+1)\geq1+. M.  (m, n)+. M.  (1, 1)\geq. M.  (m, n)+2$$

If $i\leq m$ and $j\leq n$, assume that if we compare $a_{i}$ and $b_{j}$ in $. M.  (m, n)$, the adversary's best strategy is
$1+. M\rho (p, q)+\lambda M.  (s, t)$, where $\rho,  \lambda\in \{. , /,\backslash\}$ and $s+t\geq1$ if $\lambda=.$ and $s,t\geq1$ if $\lambda=\{/,\backslash\}$. Then the adversary uses the same strategy here, and we get
\begin{displaymath}
\begin{array}{lll}
. M_{i, j}.  (m+1, n+1)&\geq1+. M\rho (p, q)+\lambda M.  (s+1, t+1)\\
&\geq1+. M\rho (p, q)+\lambda M.  (s, t)+2=. M.  (m, n)+2
\end{array}
\end{displaymath}
by applying the induction hypothesis and Lemma \ref{thm:25381}.

If $i=m+1$ and $j=n+1$,  we can handle this case as well by reversing the order of the elements in $A$ and $B$.

 Therefore  adversary can always find a strategy which results the number of comparisons not smaller than $. M.  (m, n)+2$,  no matter what the first comparison is. This  completes our proof.
\end{proof}
Now we are ready to prove  Lemma~\ref{25384}.
 \begin{proof}
We induce on $m$ and $n$.  The case for $1\leq m+n \leq 10$ are given in \cite{taocp}.
Now suppose that $m+n\geq 11$ and the lemma is already established for any $(m',n')$ satisfying $m'+n'<m+n$.

    \emph{Part (a).}

     If $m=1$, we have $.M.(1,n+1)\leq .M_{1,n+1}.(1,n+1)=\max\{1, 1+.M.(1,n)\}=1+.M.(1,n)$. If $m\geq 2$, then $.M_{m,n+1}.(m,n+1)=\max\{.M.(m-1, n+1),.M.(m, n)\}+1$, and by the induction we know $.M.  (m, n)\geq.M.  (m-1, n)+1\geq.M.(m-1, n+1)$, thus
    $$.M.(m, n+1)\leq.M_{m,n+1}.(m,n+1)\leq\max\{.M.(m-1, n+1), .M.(m, n)\}+1=.M.(m, n)+1.$$

    \emph{Part (b).}

    When $m=n$, $.M.(m+1,m)\geq .M.(m,m-1)+2\geq.M.(m-1,m-2)+4\geq\cdots\geq.M.(2,1)+2m-2=2m$ according to Lemma~\ref{25383}, thus $.M.(m+1,m)\geq.M.(m,m)+1$, since $.M.(m,m)\leq 2m-1$. When $m<n$,  we have
    $$.M.  (m+1, n)\geq.M.  (m, n-1)+2\geq.M.  (m, n)+1.$$
     The first inequality is due to Lemma~\ref{25383}, and the second one is by the induction hypothesis.
  \end{proof}

We can show a similar statement about $/M.$ function as well. The proof is very similar, and we give it as well for sake of completeness.
\begin{lemma}\label{thm:25381}
For any $m,n\geq1$, we have
\begin{enumerate}
  \item[(a)] $/M.  (m+1, n+1)\geq/M.  (m, n)+2$ \cite{yao};
\end{enumerate}
and for any $m,n\geq 1$ and $m\leq n$, we have
\begin{enumerate}
  \item[(b)]  $/M.  (m, n+1)\leq/M.  (m, n)+1$;
  \item[(c)] $/M.  (m+1, n)\geq/M.  (m, n)+1$,  except $ (m, n)= (1, 1), (2, 2)$ or  $ (3, 3)$.
\end{enumerate}
\end{lemma}
 \begin{proof}
We induce on $m$ and $n$.  The case for $1\leq m+n \leq 10$ are given in \cite{taocp}.
Now suppose that $m+n\geq 11$ and the lemma is already established for any $(m',n')$ satisfying $m'+n'<m+n$.

    \emph{Part (b).}

     If $m=1$, we have $/M.(1,n+1)\leq /M_{1,n+1}.(1,n+1)=\max\{1, 1+/M.(1,n)\}=1+/M.(1,n)$. If $m\geq 2$, then $/M_{m,n+1}.(m,n+1)=\max\{/M.(m-1, n+1),/M.(m, n)\}+1$, and by the induction we know $/M.  (m, n)\geq/M.  (m-1, n)+1\geq/M.(m-1, n+1)$, thus
    $$/M.(m, n+1)\leq/M_{m,n+1}.(m,n+1)\leq\max\{/M.(m-1, n+1), /M.(m, n)\}+1=/M.(m, n)+1.$$

    \emph{Part (c).}

    When $m=n$, $/M.(m+1,m)\geq /M.(m,m-1)+2\geq/M.(m-1,m-2)+4\geq\cdots\geq/M.(5,4)+2m-8=2m$ according to \emph{Part (a)}, thus $/M.(m+1,m)\geq/M.(m,m)+1$, since $/M.(m,m)\leq 2m-1$. When $m<n$,  we have
    $$/M.  (m+1, n)\geq/M.  (m, n-1)+2\geq/M.  (m, n)+1.$$
     The first inequality is due to \emph{Part(a)}, and the second one is by the induction hypothesis.
  \end{proof}


\section {Lower bounds for $\alpha (m)$}\label{section:lowerbound}

The key step is to show that  $. M.  (m+25, n+38)\geq. M.  (m, n)+63$, which directly implies Theorem~\ref{thm:38/25}.
Since it's unavoidable to show similar statements for other restricted adversaries $\lambda M \rho$, we prove them by induction in parallel. In addition, by symmetry, we have  $/M.(m,n)=.M\backslash(m,n)$, $\backslash M.(m,n)=.M/(m,n)$,  and $/M/(m,n)=\backslash M \backslash(m,n)$ , so the following theorem is enough for our goal. The idea is similar with Lemma~\ref{25383}. We suggest readers to read the proof the Lemma~\ref{25383} at first as a warmup.
\begin{theorem}\label{2538}
For $m,n\geq0$ and $m+n\geq1$, we have
\begin{enumerate}
   \item[(a)] $. M.  (m+25, n+38)\geq. M.  (m, n)+63$;
\end{enumerate}
and for $m,n\geq 1$, we have
\begin{enumerate}
   \item[(b)] $/M.  (m+25, n+38)\geq/M.  (m, n)+63$;
   \item[(c)] $\backslash M.  (m+25, n+38)\geq\backslash M.  (m, n)+63$;
   \item[(d)] $/M\backslash (m+25, n+38)\geq/M\backslash (m, n)+63$;
   \item[(e)] $\backslash M \backslash (m+25, n+38)\geq\backslash M\backslash (m, n)+63$, except $(m, n)=(1, 1)$;
   \item[(f)] $\backslash M/ (m+25, n+38)\geq \backslash M/ (m, n)+63$, except $(m, n)=(2, 1)$.
\end{enumerate}
\end{theorem}

\begin{proof}

The proof is by induction on $m$ and $n$. The starting values for $ m, n \leq 50$  are given in \cite{table}.  Now suppose the theorem holds for any $m',n'$ satisfying $m'\leq m$, $n'\leq n$ and $m'+n'<m+n$, we then prove the case $(m,n)$ where $m\geq 51$ or $n\geq 51$. Recall that our task is to design a strategy for the adversary for $(m+25,n+38)$.

\emph{\textbf{Part (a).}}
Suppose an algorithm begins by comparing $a_{i}$ and $b_{j}$,  if $i\leq m$ and $j\geq n+1$,  then the adversary claims $a_{i}<b_{j}$ and follows the simple strategy,  yielding
$$. M_{i, j}.  (m+25, n+38)\geq1+. M.  (m, n)+. M.  (25, 38)=. M.  (m, n)+63.$$

If $i\geq m+1$ and $j\leq n$, the adversary claims $a_{i}>b_{j}$ and uses the simple strategy. This leads to
$$. M_{i, j}.  (m+25, n+38)\geq1+. M.  (m, n)+. M.  (25, 38)=. M.  (m, n)+63.$$

If $i\leq m$ and $j\leq n$,  assume if we compare $a_{i}$ and $b_{j}$ in $. M.  (m, n)$,  the adversary's best strategy is
$1+. M\rho (p, q)+\lambda M.  (s, t)$ where $\lambda,  \rho\in\{.,/,\backslash$\},  then adversary uses the same strategy here, and we get
\begin{displaymath}
\begin{array}{lll}
. M_{i, j}.  (m+25, n+38)&\geq 1+. M\rho (p, q)+\lambda M.  (s+25, t+38)\\
&\geq1+. M\rho (p, q)+\lambda M.  (s, t)+63\\
&\geq . M.  (m, n)+63
\end{array}
\end{displaymath}
by using the induction hypothesis.

If $i\geq m+1$ and $j\geq n+1$,  then $i\leq25$ and $j\leq38$ cannot happen simultaneously, thus there are only three possible cases: ($i\geq26$, $j\leq38$), ($i\leq25$, $j\geq39$), or ($i\geq26$, $j\geq39$). Reversing the order of the elements in $A$ and $B$ maps all these three cases to the above ones,
so we can handle these cases as well by symmetry.

Therefore no matter which two elements are chosen to compare at the first step, the adversary can always find a strategy resulting the value not smaller than $. M.  (m, n)+63$. This completes the proof of \emph{Part (a)}.

\emph{\textbf{Part (b).}}
Suppose an algorithm begins by comparing $a_{i}$ and $b_{j}$,  if $i\leq m$ and $j\geq n+1$,  then the adversary claims $a_{i}<b_{j}$ and follows the simple strategy,  yielding
$$/M_{i, j}.  (m+25, n+38)\geq1+/M.  (m, n)+. M.  (25, 38)=/M.  (m, n)+63.$$

If $i\geq m+1$ and $j\leq n$, the adversary claims $a_{i}>b_{j}$ and uses the simple strategy. This leads to
$$/M_{i, j}.(m+25, n+38)\geq1+/M.  (m, n)+. M.  (25, 38)=/M.  (m, n)+63.$$

If $i\leq m$ and $j\leq n$,  assume if we compare $a_{i}$ and $b_{j}$ in $/M.  (m, n)$,  the adversary's best strategy is
$1+/M\rho (p, q)+\lambda M.  (s, t)$ where $\lambda,  \rho\in\{.,/,\backslash$\},  then adversary uses the same strategy here, and we get
\begin{displaymath}
\begin{array}{lll}
/M_{i, j}. (m+25, n+38)&\geq1+/M\rho (p, q)+\lambda M.  (s+25, t+38)\\
&\geq1+/M\rho (p, q)+\lambda M.  (s, t)+63\\
&\geq/M.  (m, n)+63
 \end{array}
 \end{displaymath}
 by using the induction hypothesis.

If $i\geq m+1$ and $j\geq n+1$, then $i\leq25$ and $j\leq38$ cannot happen simultaneously, so we only need to consider the following cases:

 If $i\leq25$ and $j\geq39$, or $i\geq26$ and $j\leq38$, the adversary uses the simple strategy, yielding
$$/M_{i,j}.  (m+25, n+38)\geq1+/M.  (25, 38)+. M.  (m, n)\geq1+62+.M.(m,n)\geq/M.  (m, n)+63.$$

 If $i\geq26$ and $j\geq39$:  assume if we compare $a_{i-25}$ and $b_{j-38}$ in $/M.  (m, n)$,  the adversary's best strategy is
$1+/M\rho (p, q)+\lambda M.  (s, t)$. If $ (p, q, \rho)\neq  (1, 1, /)$,  the adversary uses the same strategy,  and we get
\begin{displaymath}
\begin{array}{lll}
/M_{i, j}.  (m+25, n+38)&\geq1+/M\rho (p+25, q+38)+\lambda M.  (s, t)\\
&\geq1+/M\rho (p, q)+63+\lambda M.  (s, t)\\
&\geq/M.  (m, n)+63
 \end{array}
 \end{displaymath}
 by the induction hypothesis.
If $ (p, q, \rho)= (1, 1, /)$,  then we know that $i\geq 27$ and $j=39$, and  the adversary can use the simple strategy, yielding
\begin{displaymath}
\begin{array}{lll}
/M_{i, j}.  (m+25, n+38)&\geq1+/M.  (25, 39)+. M.  (m, n-1)\\
&=1+63+.M.(m,n-1)\geq63+/M.  (m, n).
 \end{array}
 \end{displaymath}
 The last inequality is due to Lemma \ref{prop:yao}.

Therefore the adversary can always find a strategy resulting the value not smaller than $/M.  (m, n)+63$, no matter what the first comparison is. This completes the proof of \emph{Part (b)}.

\emph{\textbf{Part (c).}}
Suppose an algorithm begins by comparing $a_{i}$ and $b_{j}$,  if $i\leq m$ and $j\geq n+1$,  then the adversary claims $a_{i}<b_{j}$ and follows the simple strategy,  yielding
$$\backslash M_{i, j}.  (m+25, n+38)\geq1+\backslash M.  (m, n)+. M.  (25, 38)=\backslash M.  (m, n)+63.$$

If $i\geq m+1$ and $j\leq n$, the adversary claims $a_{i}>b_{j}$ and uses the simple strategy. This leads to
$$\backslash  M_{i, j}.  (m+25, n+38)\geq1+\backslash  M.  (m, n)+. M.  (25, 38)=\backslash  M.  (m, n)+63.$$

If $i\leq m$ and $j\leq n$,  assume if we compare $a_{i}$ and $b_{j}$ in $. M.  (m, n)$,  the adversary's best strategy is
$1+\backslash M\rho (p, q)+\lambda M.  (s, t)$ where $\lambda,  \rho\in\{.,/,\backslash$\},  then adversary uses the same strategy here, and we get
\begin{displaymath}
\begin{array}{lll}
\backslash  M_{i, j}.  (m+25, n+38)&\geq1+\backslash M\rho (p, q)+\lambda M.  (s+25, t+38)\\
&\geq1+\backslash  M\rho (p, q)+\lambda M.  (s, t)+63\\
&\geq. M.  (m, n)+63
\end{array}
\end{displaymath}
  by using the induction hypothesis.

Similar with the above argument, if $i\geq m+1$ and $j\geq n+1$,  we only need to investigate the following cases:

 If $i\leq25$ and $j\geq39$,  or $i\geq27$ and $j\leq38$,   the adversary uses the complex strategy with $a_{26}$ in both subproblems. This leads to
 \begin{displaymath}
\begin{array}{lll}
\backslash M_{i, j}.  (m+25, n+38)&\geq1+\backslash M/ (26, 38)+\backslash M.  (m, n)\\
&=1+62+\backslash M.  (m, n)=\backslash M.  (m, n)+63.
\end{array}
\end{displaymath}

 If $i=26$ and $j\leq38$, since $i\geq m+1$ and $j\geq n+1$, then $m\leq25$ and $n\leq37$ and these cases have been checked as starting values.

 If $i\geq26$ and $j\geq39$: assume if we compare $a_{i-25}$ and $b_{j-38}$ in $\backslash M.  (m, n)$,  the adversary's best strategy is
$1+\backslash M\rho (p, q)+\lambda M.  (s, t)$.  If $ (p, q, \rho)\neq  (1, 1, \backslash)$ or $(2, 1, /)$,  the adversary  uses the same strategy, yielding
\begin{displaymath}
\begin{array}{lll}
\backslash M_{i, j}.  (m+25, n+38)&\geq1+\backslash M\rho (p+25, q+38)+\lambda M.  (s, t)\\
&\geq1+\backslash M\rho (p, q)+63+\lambda M.  (s, t)\\
&\geq \backslash M.  (m, n)+63
\end{array}
\end{displaymath}
 by the induction hypothesis.
If $ (p, q, \rho)= (1, 1, \backslash)$,  we have $i=26$ and $j\geq40$.  The adversary claims $a_{i}<b_{j}$ and  follows the simple strategy, yielding
\begin{displaymath}
\begin{array}{lll}
\backslash M_{i, j}.  (m+25, n+38)&\geq1+\backslash M.  (26, 38)+. M.  (m-1, n)\\
&=1+63+. M.  (m-1, n)\geq63+\backslash M. (m, n)
\end{array}
\end{displaymath}
by using Lemma \ref{prop:yao}.
If $ (p, q, \rho)= (2, 1, /)$,  we have $i\geq28$ and $j=39$ or $i=26$ and $j>39$, since the case where $i=26$ and $j>39$ has already been considered,   we only need to investigate the case where $i\geq28$ and $j=39$. Notice that $j\geq n+1$, i.e.  $n \leq 38$,  hence $m\geq 50$.
If $i>28$,  the adversary claims $a_{i}>b_{j}$ and follows the complex strategy with $a_{28}$ in both subproblems. This leads to
\begin{displaymath}
\begin{array}{lll}
\backslash M_{i, j}.  (m+25, n+38)&\geq1+\backslash M/ (28, 39)+\backslash M.  (m-2, n-1)\\
&=1+66+\backslash M.  (m-2, n-1)\\
&\geq66+\backslash M.  (m-1, n-1)\\
&\geq63+1+\backslash M/ (2, 1)+\backslash M.  (m-1, n-1)\\
&=63+\backslash M_{i-25, 1}.  (m, n)\geq\backslash M.  (m, n)+63.
\end{array}
\end{displaymath}
 The second inequality is according to Lemma \ref{thm:25381} and the second equality is the assumption of the best strategy for the adversary.
 If $i=28$ and $j=39$, we have $m\leq 27$ and $n\leq 38$, which have been checked as starting values.

Therefore the adversary can always find a strategy resulting the value not smaller than $\backslash M.  (m, n)+63$. This completes the proof of \emph{Part (c)}.

%
 \emph{\textbf{Part (d).}}
 If $i\leq25$ and $j\geq39$, or $i\geq26$ and $j\leq38$, the adversary uses the simple strategy, yielding
 \begin{displaymath}
\begin{array}{lll}
/M_{i,j}\backslash  (m+25, n+38)&\geq1+/M.  (25, 38)+. M\backslash  (m, n)\\
&\geq1+62+.M\backslash(m,n)\geq/M\backslash  (m, n)+63.
\end{array}
\end{displaymath}

 If $i\geq26$ and $j\geq39$:  assume if we compare $a_{i-25}$ and $b_{j-38}$ in $/M\backslash  (m, n)$,  the adversary's best strategy is
$1+/M\rho (p, q)+\lambda M\backslash  (s, t)$. If $ (p, q, \rho)\neq  (1, 1, /)$,  the adversary uses the same strategy,  and we get
\begin{displaymath}
\begin{array}{lll}
/M_{i, j}\backslash  (m+25, n+38)&\geq1+/M\rho (p+25, q+38)+\lambda M\backslash  (s, t)\\
&\geq1+/M\rho (p, q)+63+\lambda M\backslash (s, t)\\
&=/M\backslash  (m, n)+63
\end{array}
\end{displaymath}
by the induction hypothesis.
If $ (p, q, \rho)= (1, 1, /)$,  then we know that $i\geq 27$ and $j=39$, and  the adversary can use the simple strategy, yielding
\begin{displaymath}
\begin{array}{lll}
/M_{i, j}\backslash (m+25, n+38)&\geq1+/M.  (25, 39)+. M\backslash (m, n-1)\\
&\geq1+63+.M\backslash(m,n-1)\geq63+/M\backslash(m, n).
\end{array}
\end{displaymath}
The last inequality is due to Lemma \ref{prop:yao}.

 If $i\leq 25$ and $j\leq 38$, reserving the order of the elements maps this case to the above cases, thus we can handle this case as well by symmetry.

Therefore the adversary can always find a strategy which is not smaller than $/M\backslash (m, n)+63$, no matter which the first comparison is. So \emph{Part (d)} is true.

\emph{\textbf{Part (e).}}
 If $i\geq m+1$ and $j\leq n$, or $i\leq m$ and $j\geq n+1$, the adversary uses the simple strategy, yielding
$$\backslash M_{i,j}\backslash  (m+25, n+38)\geq1+\backslash M.(m,n)+.M\backslash  (25, 38)=\backslash M\backslash  (m, n)+63.$$
 If $i\leq m$ and $j\leq n$:  assume if we compare $a_{i}$ and $b_{j}$ in $\backslash M\backslash  (m, n)$,  the adversary's best strategy is
$1+\backslash M\rho (p, q)+\lambda M\backslash  (s, t)$. If $ (s, t, \lambda)\neq  (1, 1, \backslash)$,  the adversary uses the same strategy,  and we get
\begin{displaymath}
\begin{array}{lll}
\backslash M_{i, j}\backslash  (m+25, n+38)&\geq1+\backslash M\rho (p, q)+\lambda M\backslash  (s+25, t+38)\\
&\geq1+\backslash M\rho (p, q)+\lambda M\backslash (s, t)+63\\
&=\backslash M\backslash  (m, n)+63
\end{array}
\end{displaymath}
 by the induction hypothesis.
If $ (s, t, \lambda)= (1, 1, \backslash)$,  then we know that $i\leq m-1$ and $j=n$, and  the adversary can use the simple strategy, yielding
\begin{displaymath}
\begin{array}{lll}
\backslash M_{i, j}\backslash (m+25, n+38)&\geq1+\backslash M.(m, n-1)+. M\backslash (25, 39)  \\
&\geq1+.M\backslash(m,n-1)+63\geq63+\backslash M\backslash(m, n).
\end{array}
\end{displaymath}
The last inequality is due to Lemma \ref{prop:yao}.

If $i\geq m+1$ and $j\geq n+1$, then $i\leq25$ and $j\leq38$ cannot happen simultaneously, so we only need to investigate the following cases:

 If $i\leq25$ and $j>38$, or $i\geq27$ and $j\leq38$,  the adversary uses the complex strategy with $a_{26}$ in both subproblems. This leads to
$$\backslash M_{i, j}\backslash (m+25, n+38)\geq1+\backslash M/ (26, 38)+\backslash M\backslash (m, n)=\backslash M\backslash (m, n)+63.$$

 If $i=26$ and $j\leq38$, then $m\leq25$ and $n\leq38$, which have been checked as starting values.

 If $i\geq26$ and $j\geq39$: assume if we compare $a_{i-25}$ and $b_{j-38}$ in $\backslash M\backslash (m, n)$,  the adversary's best strategy is
$1+\backslash M\rho (p, q)+\lambda M\backslash (s, t)$. If $ (p, q, \rho)\neq (1, 1, \backslash)$ or  $(2, 1, /)$,  the adversary uses the same strategy, yielding
$$\backslash M_{i, j}\backslash (m+25, n+38)\geq1+\backslash M\rho (p+25, q+38)+\lambda M\backslash (s, t)\geq \backslash M\backslash (m, n)+63$$ by the induction hypothesis.
If $ (p, q, \rho)= (1, 1, \backslash)$,  we have $i=26$ and $j>39$.  The adversary claims $a_{i}<b_{j}$ and uses the simple strategy, leading to
\begin{displaymath}
\begin{array}{lll}
\backslash M_{i, j}\backslash (m+25, n+38)&\geq1+\backslash M.  (26, 38)+. M\backslash (m-1, n)\\
&=64+. M\backslash (m-1, n)\geq63+\backslash M\backslash (m, n)
\end{array}
\end{displaymath}
 by using Lemma~\ref{prop:yao}.

If $ (p, q, \rho)= (2, 1, /)$,  the case where $i\geq28$ and $j=39$ is the only unconsidered case.
If we compare $a_{2}$ with $b_{1}$ in $\backslash M\backslash (m, n)$,  and the adversary claims $a_{2}>b_{1}$ , then the best strategy must be $1+\backslash M.  (1, 1)+. M\backslash (m-1, n-1)$, otherwise the adversary claims $a_{2}<b_{1}$ , and the best strategy must be $1+. M\backslash (m-2, n)$.  So we get $\backslash M\backslash (m, n)\leq \backslash M_{2, 1}\backslash (m, n)= \max\{1+. M\backslash (m-2, n), 2+. M\backslash (m-1, n-1)\}$.  If $. M\backslash (m-2, n)\geq 1+. M\backslash (m-1, n-1)$, then $.M\backslash(m-2,n)+1\geq\backslash M\backslash(m,n)$, and the adversary splits the problem $\backslash M\backslash (m+25, n+38)$ into two independent subproblems $\backslash M.  (27, 38)$ and $. M\backslash (m-2, n)$ before the algorithm begins, and this leads to
$$\backslash M\backslash (m+25, n+38)\geq\backslash M.  (27, 38)+. M\backslash (m-2, n)\geq63+\backslash M\backslash (m, n)$$
 Otherwise ($. M\backslash (m-2, n)<1+. M\backslash (m-1, n-1)$), then $2+.M\backslash(m-1,n-1)\geq\backslash M\backslash(m,n)$ and the adversary claims $a_{i}>b_{j}$ and uses the simple strategy, yielding

$$\backslash M\backslash (m+25, n+38)\geq\backslash M.  (26, 39)+. M\backslash (m-1, n-1)+1\geq\backslash M\backslash (m, n)+63.$$

Therefore the adversary can always find a strategy resulting the value not smaller than $\backslash M\backslash  (m, n)+63$. This completes the proof of \emph{Part (e)}.


\emph{\textbf{Part (f).}}
If $n>50$,  we can assume $j\geq\lfloor\frac{38+n}{2}\rfloor\geq 44$ by symmetry.
 If $i\leq25$,  the adversary claims $a_{i}<b_{j}$ and uses the complex strategy with $a_{26}$ in both subproblems. This leads to
 $$\backslash M/ (m+25, n+38)\geq 1+\backslash M/ (26, 38)+\backslash M/ (m, n)= 63+\backslash M/ (m, n).$$

 If $i\geq26$:  assume that if we compare $a_{i-25}$ and $b_{j-38}$ in $\backslash M\backslash (m, n)$, the adversary's best strategy is
$1+\backslash M\rho (p, q)+\lambda M/ (s, t)$.  If $ (p, q, \rho)\neq  (1, 1, \backslash),  (2, 1, /)$,   adversary uses the same strategy, yielding
$$\backslash M_{i, j}/ (m+25, n+38)\geq1+\backslash M\rho (p+25, q+38)+\lambda M/ (s, t)\geq \backslash M/ (m, n)+63$$ by using the induction hypothesis.
If $ (p, q, \rho)= (1, 1, \backslash)$,  we have $i=26$ and $j\geq40$.  The adversary claims $a_{i}<b_{j}$  and use the simple strategy, yielding
\begin{displaymath}
\begin{array}{lll}
\backslash M_{i, j}/ (m+25, n+38)&\geq1+\backslash M.  (26, 28)+. M/ (m-1, n)\\
&=64+. M/ (m-1, n)\geq63+\backslash M/ (m, n)
\end{array}
\end{displaymath}
by using Lemma \ref{prop:yao}.
If $ (p, q, \rho)= (2, 1, /)$,  we have $i\geq28$ and $j=39$, violating the assumption  $j\geq44$.

If $n\leq50$,  then $m>50$ and we can assume $i\geq\lfloor\frac{25+m}{2}\rfloor\geq 38$ by symmetry.
  If $j\leq38$,  the adversary claims $a_{i}>b_{j}$,  and uses the complex strategy with $a_{26}$ in both subproblems, yielding
 $$\backslash M/ (m+25, n+38)\geq 1+\backslash M/ (26, 38)+\backslash M/ (m, n)\geq 63+\backslash M/ (m, n).$$

 If $j\geq39$: assume if we compare $a_{i-25}$ and $b_{j-38}$ in $\backslash M\backslash (m, n)$,  the adversary's best strategy is
$1+\backslash M\rho (p, q)+\lambda M/ (s, t)$.  If $ (p, q, \rho)\neq  (1, 1, \backslash)$ or  $(2, 1, /)$,  then adversary uses the same strategy, thus
$$\backslash M_{i, j}/ (m+25, n+38)\geq1+\backslash M\rho (p+25, q+38)+\lambda M/ (s, t)\geq \backslash M/ (m, n)+63$$ by using the induction hypothesis.
If $ (p, q, \rho)= (2, 1, /)$,  we have $j=39$.
Similar with the argument in \emph{Part (e)}, we get $\backslash M/ (m, n)\leq \backslash M_{2, 1}/ (m, n)= \max\{1+. M/ (m-2, n), 2+. M/ (m-1, n-1)\}$.  If $. M/ (m-2, n)\geq 1+. M/ (m-1, n-1)$,  the adversary splits the problem $\backslash M/ (m+25, n+38)$ into two independent subproblems $\backslash M.  (27, 38)$ and $. M/ (m-2, n)$ before the first comparison begins,  and this leads to
$$\backslash M/ (m+25, n+38)\geq\backslash M.  (27, 38)+. M/ (m-2, n)\geq63+\backslash M/ (m, n).$$
Otherwise ($. M/ (m-2, n)<1+. M/ (m-1, n-1)$),  then the adversary claims $a_{i}>b_{j}$,  and uses the simple strategy, yielding
\begin{displaymath}
\begin{array}{lll}
\backslash M/ (m+25, n+38)&\geq\backslash M.  (26, 39)+. M/ (m-1, n-1)+1\\
&\geq65+. M/ (m-1, n-1)\geq\backslash M/ (m, n)+63.
\end{array}
\end{displaymath}
If $ (p, q, \rho)= (1, 1, \backslash)$,  then $i=26$, violating the assumption $i\geq38$.

Therefore the adversary can always find a strategy resulting the value not smaller than $\backslash M/ (m, n)+63$. This completes the proof of \emph{Part (f)}.
\end{proof}

Now, we are ready to prove  Theorem~\ref{thm:38/25}.

\begin{proof}
   The small cases $1\leq m \leq25$ and $1\leq n \leq38$ are given in \cite{table}.
   Given any pair $ (m, n)$ satisfying $m\leq n\leq \frac{38}{25}m$,  let $m=25p+s$,  $n=38q+t$ where $0<s\leq25$, $0<t\leq38$, and observe that $m\geq 25q+\lceil\frac{25}{38}t\rceil$, thus
   \begin{displaymath}
   \begin{array}{lll}
   \overline{.M.}(m,n)=\overline{.M.}(25p+s,38q+t)\leq\overline{.M.}(25q+\lceil\frac{25}{38}t\rceil,38q+t)\leq\overline{.M.}(\lceil\frac{25}{38}t\rceil,t)=0.
   \end{array}
   \end{displaymath}
  The first inequality is due to Lemma \ref{25384} and the second one is due to Theorem \ref{2538}.
%
%
\end{proof}

\section{Limitations of Knuth's adversary methods}\label{section:upperbound2}
In this section, we prove Theorem~\ref{thm:9/5}, which shows Knuth's adversary methods can not provide lower bounds beyond $\alpha(m)\geq9\lceil m/5\rceil$.
Actually, we prove a stronger result:
\begin{theorem}\label{thm:5k9k}
     $. M.  (5k, 9k+12t)\leq14k+11t-2$,  for $k,t\geq0$ and $t+k\geq1$.
\end{theorem}
With this theorem, Theorem~\ref{thm:9/5} is obvious, since if $n\geq9\lceil m/5\rceil$,
$\overline{.M.}(m,n)\geq\overline{.M.}(m,9\lceil m/5\rceil)\geq1$.
\begin{proof}
The proof is by induction on $k$ and $t$. We verify the case $k\leq10$ first:
when $t\geq k/10+2/5$,  we have
$$. M.  (5k, 9k+12t)\leq M (5k, 10k+12t-k)\leq M_{bm} (5k, 10k+12t-k)\leq14k+11t-2.$$
When $t<k/10+2/5$,  these finite cases can be checked in \cite{table}.

Now suppose $k\geq 11$ and we have already proven this theorem for any  $(k',t')$ satisfying $k'<k$,  or $k'=k$ and $t'<t$.
 Since $. M.  (m, n)=min_{i, j}. M_{i, j}.  (m, n)\leq. M_{50, 79}.  (m, n)$,  thus it's enough to show
 $. M_{50, 79}.  (5k, 9k+12t)\leq14k+11t-2$.
 In other word, an algorithm which begins by comparing $a_{50}$ with $b_{79}$ can "beat" the adversary. We'll prove it by enumerating the adversary's best strategy.

 \emph{\textbf{Case(a).}}  The adversary claims $a_{50}<b_{79}$  and follows three possible strategies.

 \emph{(i).} The adversary uses the simple strategy, then
 $$.M_{50, 79}.(5k, 9k+12t)=1+.M.(50+x, 78-y)+.M.(5k-50-x, 9k+12t-78+y),$$
where $x, y\geq0$. Thus it's sufficient to show
 $$\overline{.M.}(50+x, 78-y)+\overline{.M.}(5k-50-x, 9k+12t-78+y)\geq \overline{.M.}(5k-50, 9k-78+12t)\geq t+2.$$
  The first inequality is according to Lemma~\ref{25384} and the second one is by the induction hypothesis.

\emph{(ii).}  The adversary uses the complex strategy,  with $a_{51+x}$ in both subproblems.
 \begin{displaymath}
\begin{array}{lll}
 . M_{50, 79}.  (5k, 9k+12t)&=1+.M/(51+x, 78-y)+\backslash M.(5k-50-x, 9k+12t-78+y)\\
 &\leq 1+.M.(51+x, 78-y)+. M.  (5k-50-x, 9k+12t-78+y),
 \end{array}
\end{displaymath}
 where $x,y\geq0$. Thus it's equivalent to show
 $$\overline{.M.}(51+x,78-y)+\overline{.M.}(5k-50,9k+12t-78+y)-1\geq \overline{.M.}(5k-50, 9k-78+12t)-1\geq t+1.$$

\emph{(iii).} The adversary uses the complex strategy  with $b_{78-y}$ in both subproblems.
\begin{displaymath}
\begin{array}{lll}
. M_{50, 79}.  (5k, 9k+12t)&=1+.M\backslash(50+x, 78-y)+/M.(5k-50-x, 9k+12t-77+y)\\
&\leq 1+.M.(50+x, 78-y)+.M.(5k-50-x, 9k+12t-77+y),
\end{array}
\end{displaymath}
 where $x, y\geq0$. Thus it's equivalent to show
 $$\overline{.M.}(50+x, 78-y)+\overline{.M.}(5k-50-x, 9k+12t-77+y)-1\geq \overline{.M.}(5k-50, 9k-78+12t)-1\geq t+1.$$

 \emph{\textbf{Case(b).}}   The adversary claims $a_{50}>b_{79}$ and follows  three possible strategies.

 \emph{(i).}  The adversary uses the simple strategy, then
 $$. M_{50, 79}.  (5k, 9k+12t)=1+. M.  (49-x, 79+y)+. M.  (5k-50+1+x, 9k+12t-79-y),$$
 where $x, y\geq0$. Thus it's sufficient to show
  $$\overline{.M.}(49-x, 79+y)+\overline{.M.}(5k-50+1+x, 9k+12t-79-y)\geq t+2.$$
 Let $5p\leq x\leq 5p+4$ and $12q-10\leq y\leq12q+1$, then we claim that $\overline{.M.}(49-x,79+y)\geq p+q+1$.
If $q\leq 2$, these finite cases can be checked in \cite{table}.
Otherwise ($q\geq 3$), then $79+y>2\times (49-5p)$, and $. M. (49-5p, 81+12(q-1))\leq M_{bm} (49-5p, 81+12(q-1))\leq127+11(q-1)-6p$.
Therefore $\overline{.M.}(49-x,79+y)\geq \overline{.M.}(49-5p,81+12(q-1))\geq1+p+q$ due to Lemma~\ref{25384}.

 Since $\overline{.M.}(49-x, 79+y)\geq p+q+1$, if $p+q\geq t+1$, we've done.
 If $p+q\leq t$, according to Lemma~\ref{25384} and the induction hypothesis, we have
 $$\overline{.M.}(5k-50+1+x, 9k+12t-79-y)\geq \overline{.M.}(5k-50+5p+5, 9k-90+9p+9+12(t-q-p))$$
 $$\geq t-p-q+1.$$

 Thus $\overline{.M.}(49-x, 79+y)+\overline{.M.}(5k-50+1+x, 9k+12t-79-y)\geq t+2$.

 \emph{(ii).}  The adversary uses the complex strategy with $a_{49-x}$ in both subproblems, then
   \begin{displaymath}
\begin{array}{lll}
 . M_{50, 79}.  (5k, 9k+12t)&=1+.M/(49-x, 79+y)+\backslash M.  (5k-50+2+x, 9k+12t-79-y)\\
 &\leq 2+. M.  (49-x, 79+y)+1+. M.  (5k-50+1+x, 9k+12t-79-y)\\
 &\leq 14k+11t-2,
  \end{array}
\end{displaymath}
 where $x, y\geq0$.

\emph{(iii).}  The adversary uses the complex strategy with $b_{80+y}$ in both subproblems, then
 \begin{displaymath}
\begin{array}{lll}
 . M_{50, 79}.  (5k, 9k+12t)&=1+.M\backslash(49-x, 80+y)+/M.(5k-50+1+x, 9k+12t-79-y)\\
 &\leq 2+. M.  (49-x, 79+y)+. M.  (5k-50+1+x, 9k+12t-79-y)\\
 &\leq 14k+11t-2,
 \end{array}
\end{displaymath}
 where $x,y\geq0$.
\end{proof}

\section{Upper bounds for $\alpha(m)$}\label{section:upper_bound1}


 In this section, we give better upper bounds for $\alpha(m)$ by proposing a simple procedure. This procedure only involves the first two elements in each $A$ and $B$ and can be viewed as a modification of binary merge.
  \begin{algorithm}
  \caption{Modified Binary Merge}
  \begin{algorithmic}
  \STATE Compare $a_{1}$ and $b_{2}$

    \IF{$a_{1}>b_{2}$}
     \STATE  merge $ (m, n-2)$.

        \ELSE
        \STATE compare $a_{2}$ and $b_{2}$

              \IF {$a_{2}>b_{2}$}
              \STATE compare $a_1$ and $b_1$, then merge $ (m-1, n-2)$.

              \ELSE
              \STATE compare $a_{2}$ and $b_{1}$

                   \IF {$a_{2}>b_{1}$}
                   \STATE compare $a_1$ and $b_1$, then merge $ (m-2, n-1)$ .

                   \ELSE
                   \STATE merge $ (m-2, n)$.
         \ENDIF
         \ENDIF
         \ENDIF
\end{algorithmic}
\end{algorithm}

It is easy to see that this procedure induces the following recurrence relation:
$$M (m, n)\leq \max\{M (m, n-2)+1, M (m-1, n-2)+3, M (m-2, n)+3, M (m-2, n-1)+4\}. $$

In the following, we'll use the induction to give  better upper bounds for $n\in[2m-2,3m]$.
The following proofs are very similar, but we give all the details for sake of completeness.
\begin{theorem}\label{thm:2m2k}
   $M (m, 2m+2k)\leq 3m+k-2$, for $m\geq 3$ and $k\geq -1$.
\end{theorem}
\begin{proof}
  We induce on $k$ and $m$. The case for $k=-1$ just follows tape merge algorithm. The case for $m=3$ are given by Hwang~\cite{hwang1980optimal} and Murphy~\cite{Murphyreport}.

%
%
%
%
%
%
   Now suppose that $m\geq 4$ and $k\geq 0$, and the claim has already been proven for any $(m',k')$ satisfying $m'+k'\leq m+k-1$.
   According to the procedure,  we have
   \begin{displaymath}
\begin{array}{lll}
M (m, 2m+2k)\leq \max\{&M (m, 2m+2 (k-1))+1, M (m-1, 2m+2k-2)+3,\\
&M (m-2, 2m+2k)+3, M (m-2, 2m+2k-1)+4\}
\end{array}
\end{displaymath}
  $\leq \max\{3m+k-2$  (the induction hypothesis),  $3 (m-1)+k-2+3$ (the induction hypothesis),  $3 (m-2)+1+k+3$  (binary merge),  $3 (m-2)+k+4$  (binary merge)$\}\leq3m+k-2$.
\end{proof}
\begin{theorem}\label{thm:2m2k1}
   $M (m, 2m+2k-1)\leq 3m+k-3$, for $ m\geq 5$ and $k\geq -1$.
\end{theorem}
\begin{proof}
  We induce on $k$ and $m$. The case for $k=-1$ just follows tape merge algorithm. The case for $m=5$ are given by M{\"o}nting~\cite{5n}.
  Now suppose that $m\geq 6$ and $k\geq0$, and the claim has already been proven for any $(m',k')$ satisfying $m'+k'\leq m+k-1$.
According to the procedure, we have
   \begin{displaymath}
\begin{array}{lll}
M (m, 2m+2k-1)\leq  \max\{&M(m,2m+2 (k-1)-1)+1,\\
& M(m-1,2m+2k-1-2)+3, \\
&M (m-2, 2m+2k-1)+3,\\
& M(m-2,2m+2k-2)+4\}
\end{array}
\end{displaymath}
 $\leq \max\{3m+k-3$ (the induction hypothesis),  $3 (m-1)+k-3+3$ (the induction hypothesis), $3 (m-2)+k+3$ (binary merge), $3 (m-2)+k-1+4$ (Theorem~\ref{thm:2m2k})$\}\leq3m+k-3$.
\end{proof}
\begin{theorem}\label{thm:m2m-2}
   $M (m, 2m-2)\leq 3m-4$, for $m\geq7$.
\end{theorem}
\begin{proof}
We induce on $m$. The case for $m=7$ has been verified by Smith and Lang\cite{game}. Now suppose that $m\geq8$ and the claim has already been proven for $m-1$.
According to \emph{S},  we have
   \begin{displaymath}
\begin{array}{lll}
M (m, 2m-2)\leq \max\{&M (m, 2m-4)+1, M (m-1, 2m-4)+3, \\
&M (m-2, 2m-2)+3, M (m-2, 2m-3)+4\}
\end{array}
\end{displaymath}
  $\leq \max\{3m-4$(tape merge), $3 (m-1)-4+3$ (the induction hypothesis),  $3(m-2)-1+3$ (Theorem~\ref{thm:2m2k}), $3(m-2)-2+4$ (Theorem~\ref{thm:2m2k1})$\}\leq3m-4$.
  \end{proof}

\begin{theorem}\label{thm:2m}
   $M (m, 2m)\leq 3m-3$, for $m\geq 10$.
\end{theorem}
\begin{proof}
We do the induction on $m$.
  Smith and Lang\cite{game} have verified the case for $m=10$. Now suppose that $m\geq11$ and the claim has already been proven for $m-1$.
According to the procedure,  we have
   \begin{displaymath}
\begin{array}{lll}
M (m, 2m)\leq \max\{&M (m, 2m-2)+1, M (m-1, 2m-2)+3,\\
&M (m-2, 2m)+3, M (m-2, 2m-1)+4\}
\end{array}
\end{displaymath}
$\leq \max\{3m-3$  (Theorem~\ref{thm:m2m-2}),  $3 (m-1)-3+3$ (the induction hypothesis),  $3 (m-2)+3$ (Theorem~\ref{thm:2m2k}),  $3(m-2)-1+4$ (Theorem~\ref{thm:2m2k1})$\}\leq3m-3$.
  \end{proof}
Finally, we put together the above theorems to get Theorem \ref{thm:upper_bound}.

As we can see in the proofs, if better basic cases can be provided, we can get better upper bounds by using this procedure. However, there is a barrier of this approach: if we want to show $\alpha(m)\leq 2m-k$ or $M(m,2m-k+1)<3m-k$, it's necessary to obtain $\alpha(m-1)\leq 2(m-1)-k$ or $M(m-1,2m-2-k+1)<3m-3-k$ at first, thus it is impossible to show $\alpha(m)\leq2m-\omega(1)$ via this approach.

\section{Conclusion}\label{section:conclusion}
In this paper we improve the lower bounds for $\alpha(m)$ from $\lfloor\frac{3}{2}m\rfloor+1$ to $\lfloor\frac{38}{25}m\rfloor$ via Knuth's adversary methods. We also show that it is impossible to get $\alpha(m)\geq9\lceil m/5\rceil\approx\frac{9}{5}m$ for any $m$ by using this methods.  We then design an algorithm which saves at least one comparison compared to binary merge for $2m-2\leq n\leq3m$. Specially, for the case $M(m,2m-2)$, our algorithm uses one comparison less than tape merge or binary merge, which means we can improve the upper bounds of $\alpha(m)$ by 1.  We wonder whether there exists a universal efficient algorithm to give significantly better upper bounds for $M(m,n)$ in the case $n\leq 2m$, or maybe it's intrinsically hard to compute $M$ functions since there doesn't exist general patterns or underlying structures in the corresponding decision trees.

Besides that, we are also curious about the following conjectures proposed by Knuth~\cite{taocp}:
\begin{conjecture}\label{conj1}
$M(m+1,n+1)\geq2+M(m,n)$.
\end{conjecture}
Via a similar proof with Lemma~\ref{thm:25381}, the above conjecture implies the following conjecture which has been mentioned in Section 1.
\begin{conjecture}\label{conj2}
$M(m+1,n)\geq1+M(m,n)\geq M(m,n+1)$, for $m\leq n$.
\end{conjecture}
In the attempt to prove these two conjectures,  we introduced the notation $.M^{(k)}.(m,n)$. Roughly speaking, $.M^k.(m,n)$ is the adversary which can delay $k$ steps to give the splitting strategy, and $.M.(m,n)=.M^0.(m,n)\leq.M^1.(m,n)\cdots\leq.M^{m+n-2}.(m,n)=M(m,n)$. In the case $k=0$, it is exactly Lemma~\ref{25383}, but it seems much harder even for $k=1$.

\bibliographystyle{plain}
\bibliography{merge_sort}

\end{document}